\documentclass[12pt,a4paper]{amsart}

\usepackage{latexsym,amssymb,amsmath,amsthm,amsfonts,enumerate,verbatim,xspace, exscale}

\usepackage{graphicx}
\usepackage{color,amsbsy}
 
\usepackage{hyperref}
\definecolor{darkblue}{rgb}{0,0,.5}
\hypersetup{colorlinks=true, breaklinks=true, linkcolor=darkblue, menucolor=darkblue,  urlcolor=darkblue,citecolor=darkblue}

\usepackage{setspace}

\theoremstyle{plain}
\newtheorem{theorem}{Theorem}[section]
\newtheorem{lemma}[theorem]{Lemma}
\newtheorem{proposition}[theorem]{Proposition}
\newtheorem{corollary}[theorem]{Corollary}

\theoremstyle{definition}
\newtheorem{definition}[theorem]{Definition}
\newtheorem{remark}[theorem]{Remark}

\def\d{\textup{div}}

\def\R{\mathbb{R}}

\def\F{\mathcal{F}}

\providecommand{\bysame}{\makebox[3em]{\hrulefill}\thinspace}

\newcommand{\up}{\upshape}

\newcommand{\longto}{\longrightarrow}

\def\vv<#1>{\langle#1\rangle}
\def\ww<#1>{\langle\langle#1\rangle\rangle}

\newcommand{\tr}{\mbox{$\textup{Tr}$}}

\newcommand{\ev}{\mbox{$\text{\up{ev}}$}}

\newcommand{\dd}[2]{\mbox{$\frac{\partial #2}{\partial #1}$}}
\newcommand{\ddo}{\mbox{$\dd{t}{}|_{0}$}}
\providecommand{\del}{\partial}

\newcommand{\om}{\omega}
\newcommand{\Om}{\Omega}

\newcommand{\wt}[1]{\mbox{$\widetilde{#1}$}}

\newcommand{\by}[2]{\mbox{$\frac{#1}{#2}$}}
\newcommand{\cinf}{\mbox{$C^{\infty}$}}
\providecommand{\set}[1]{\mbox{$\{#1\}$}}

\newcommand{\gu}{\mathfrak{g}}

\newcommand{\Ad}{\mbox{$\text{\upshape{Ad}}$}}
\newcommand{\ad}{\mbox{$\text{\upshape{ad}}$}}

\newcommand{\X}{\mbox{$\mathcal{X}$}}

\newcommand{\revise}[1]{#1}
\newcommand{\retwo}[1]{#1}

\title[A Hamiltonian Mean-Field System for the Navier-Stokes Equation]{%
A Hamiltonian Mean-Field System for the Navier-Stokes Equation}

\author{Simon Hochgerner}
\address{Austrian Financial Market Authority (FMA),
Otto-Wagner Platz 5, A-1090 Vienna
}
\email{simon.hochgerner@gmail.com} 
\date{September 2, 2018}

\begin{document}

\maketitle 

\begin{abstract}
We use a Hamiltonian interacting particle system to derive a stochastic mean field system whose McKean-Vlasov equation yields the incompressible Navier Stokes equation. 
Since the system is Hamiltonian, the particle relabeling symmetry implies a Kelvin Circulation Theorem along stochastic Lagrangian paths. Moreover, issues of energy dissipation are discussed and the model is connected to other approaches in the literature. 
\end{abstract}

\section*{Introduction} 
Stochastic fluid dynamics can be discussed from different perspectives:
\begin{enumerate}[\up (1)]
    \item
    Multi-scale approach, stochastic dynamics for modelling fluid flow under uncertainties:
    The idea is to separate the dynamics into a slow (deterministic) and a fast (stochastic) component. The result is a stochastic system and the goal is generally to study the corresponding S(P)DE as a realistic model of fluid motion. Representatives of this approach are \cite{CGD17,CFH17,Holm15,Mem14,RMC17,RMC17a,RMC17b}.
    \item
    Stochastic approaches to deterministic fluid mechanics:
    Again one uses a stochastic perturbation to capture fine-scale effects. But, in contrast to (1), the goal is to average over the stochastic system to gain information about (or, a solution of) the resulting deterministic model. 
    Representatives of this approach are \cite{CC07,CI05,C,CS09,E10,H17,I06b,IM08,IN11,NS17,Y83}.  
\end{enumerate}
The present paper belongs to the second category. 
We are concerned with 
the incompressible Navier-Stokes equation
on the $n$-dimensional torus $M=T^n = \R^n/\mathbb{Z}^n$:
\begin{align*}
    \dd{t}{}u = -\nabla_u u + \eta\Delta u - \nabla p,\;
    \d\,u = 0 ,\;
    u(0,x) = u_0(x)  
\end{align*}
where $\eta$ is the viscosity and $u = u(t,x)$ and $p=p(t,x)$ are Eulerian velocity and pressure, respectively.

Our approach is based on the following ``separation of dynamics'' idea: We assume that each fluid parcel consists of a large number of identical particles. The corresponding dynamics is
then derived as a stochastic Hamiltonian system with respect to an energy that consists of two components:
\begin{enumerate}[\up (1)]
    \item 
    a deterministic part due to the total momentum of the ensemble of fluid particles;
    \item
    a stochastic part due to independent random impacts on the individual particles.
\end{enumerate}
Thus the particles interact deterministically via the total ensemble momentum. 
In the limit, as the number of particles goes to infinity, this interacting particle system (IPS) tends to a mean field stochastic differential equation~\eqref{e:sdeP}. Averaging over solutions to this mean field SDE yields a deterministic PDE that describes the deterministic dynamics of the original fluid parcel at the macroscopic level.    
Essentially only assuming that the random impacts are independent and Gaussian, we show that this PDE is the incompressible Navier-Stokes equation. The incompressibility condition follows because we use the space of volume preserving diffeomorphisms as the configuration space for our system.

\subsubsection{Description of contents}
Section~\ref{sec:prel} contains notation and preliminaries. We give a detailed exposition of the Hamiltonian structure used throughout the paper. 

Section~\ref{sec:ham-ips} describes the above mentioned Hamiltonian IPS for fluid dynamics. 
In order to have a simple and tractable picture, we start with a system of identical particles on the real line. Then we use the Hamiltonian structure introduced in Section~\ref{sec:prel} to transfer the construction to the phase space of incompressible fluid mechanics.

Section~\ref{sec:main} contains our main Equation~\eqref{e:sdeP} and shows how this is obtained as a mean field limit from the IPS in Section~\ref{sec:ham-ips}. We remark that this passage to the mean field equation is carried out under the assumption that the limit exists. 
Theorem~\ref{thm:main} then shows that averaging over solutions of \eqref{e:sdeP} leads to solutions of the Navier-Stokes equation for incompressible flow. 
Thus the Navier-Stokes equation is obtained from the McKean-Vlasov equation for \eqref{e:sdeP}.
A converse is proved as-well: a Gaussian stochastic perturbation of Lagrangian trajectories, corresponding to a solution  of the Navier-Stokes equation, yields a mean-field SDE whose mean field coincides with the original solution. Due to the Hamiltonian structure we also obtain a Kelvin Circulation Theorem (Proposition~\ref{prop:KC}) that holds along stochastic Lagrangian paths. 

Section~\ref{sec:energy} is concerned with issues of energy dissipation. It is shown that the system~\eqref{e:sdeP} neither conserves stochastic energy, nor does the average over the stochastic energy dissipate. However, considering a slight modification of the equation proposed in \cite{H17}, one does obtain a stochastic energy whose average dissipates and bounds the energy of corresponding solutions to the Navier-Stokes equation.    

Section~\ref{sec:IPS} discusses energy dissipation for the empirical mean of the corresponding IPS. This is carried out, as in Section~\ref{sec:energy}, for Equation~\eqref{e:sdeP} and for the system proposed in \cite{H17}. It is found that the average over the (stochastic) energy of the empirical mean need not dissipate. 
While this is contrary to intuition, since one would expect the energy of the empirical mean to behave as that of the deterministic solution, this result is consistent with \cite{IM08,IN11,NS17}. 
 
Section~\ref{sec:comp} offers a discussion of connections with other approaches to stochastic fluid dynamics. We compare our model~\eqref{e:sdeP} to representatives of both of the categories mentioned at the beginning of this introduction.
Moreover, in Section~\ref{sec:H17} it is explained why complying with the Kelvin Circulation Theorem is a desirable property. 

Section~\ref{sec:conc} recapitulates the results and draws conclusions.

\textbf{Acknowledgements.}
I am grateful to Darryl Holm for useful remarks and explanations. 
The referee reports are also gratefully acknowledged, in particular for pointing out references \cite{Holm02a,HT12,Mem14,RMC17,RMC17a,RMC17b}.

\section{Preliminaries: phase space and Hamiltonian structure}\label{sec:prel}
Let $(\Om,\F,(\F_t)_{t\in[0,T]},P)$ be a filtered probability space satisfying the usual assumptions. In the following, all stochastic processes shall be understood to be adapted to this filtration.

\subsection{Volume preserving diffeomorphisms}
Let
$M=T^n=\R^n/\mathbb{Z}^n$. We fix $s>1+n/2$ and let $G^s$ denote the infinite dimensional $\cinf$-manifold of $H^s$-diffeomorphisms on $M$. 
Further, $G^s_0$ denotes the submanifold of volume preserving diffeomorphisms of Sobolev class $H^s$. Both, $G^s$ and $G^s_0$, are  topological groups but not  Lie groups since left composition is only continuous but not smooth. Right composition is smooth.
The tangent space of $G^s$ (resp.\ $G^s_0$) at the identity $e$ shall be denoted by $\gu^s$ (resp.\ $\gu^s_0$). 
Let $\X^s(M)$ denote the vector fields on $M$ of class $H^s$ and $\X_0^s(M)$ denote the subspace of divergence free vector fields of class $H^s$.
We have $\gu^s_0 = \X^s_{0}(M)$ and $\gu^s=\X^s(M)$.
See \cite{EM70,MEF}.
We will keep track of the superscript $s$ only when it is needed for clarification. Otherwise, we shall assume $s$ to be fixed throughout and write $G$, $G_0$, $\gu$, $\gu_0$ instead of $G^s$, $G^s_0$, $\gu^s$, $\gu^s_0$.  

We use right multiplication $R^g: G_0\to G_0$, $k\mapsto k\circ g = kg$ to trivialize the tangent bundle $TG_0\cong G_0\times\gu_0$, $\xi_g\mapsto(g,(TR^g)^{-1}\xi_g)$. 

By right multiplication we extend the $L^2$ inner product $\ww<.,.>$ on $\gu_0$ to a Riemannian metric $\mu$ on $G_0$, that is
\[
  \mu_g(\xi_g,\eta_g) 
  = \int_M\vv<(TR^g)^{-1}\xi_g(x),(TR^g)^{-1}\eta_g(x)>\, dx
  = \ww<(TR^g)^{-1}\xi_g, (TR^g)^{-1}\eta_g>
\]
for $\xi_g,\eta_g\in T_gG$, where $dx$ is the standard volume element in $M$.
We shall subsequently use $\vv<.,.>$ to denote both, the standard inner product in $\R^n$ and the duality pairing~\eqref{e:dual-p}.  

\subsection{Leray-Hodge projection $P$}
The Leray-Hodge projection operator is defined as follows (see \cite[Corollary 1.4.4]{MEF}):
Consider an $H^s$ vector field $\xi$ on $M$. Then there is a unique divergence free vector field $\eta$ of class $H^s$ and a function $f$ on $M$ such that $\xi=\eta+\nabla f$. Setting $\eta=P\xi$ thus defines a bounded linear operator $P: \X^s(M)\to\X_0^s(M)$ (for arbitrary $s\ge0$). Further, $\ww<\eta,\nabla f> = 0$.

\subsection{Inertia operator}\label{sec:InOp}
Let $\gu_0'$ denote the space of continuous linear functionals on $\gu_0$. Since $\mu$ is a weak Riemannian metric, we do not obtain an isomorphism between $\gu_0$ and $\gu_0'$. Following \cite{AK98}, define the smooth dual $\gu_0^*$ as  
\[
 \gu_0^* = \Om^1(M)^s/\set{\psi\in\Om^1(M)^s: \psi = df}
\]
which is the space of $1$-forms of Sobolev class $H^s$ modulo exact $1$-forms of class $H^s$. 
The \emph{inertia tensor} is defined as the isomorphism 
\begin{equation}
    \label{equ:I}
    \check{\mu}: \gu_0\longto\gu_0^*,
    \quad
    \xi\longmapsto[\xi^{\flat}]
\end{equation}
where $\xi^{\flat}$ is the $1$-form associated to $\xi$ (raising indices)  and $[\cdot]$ denotes the equivalence class. The inverse of $\check{\mu}$ is given by $\check{\mu}^{-1}[\psi+df] = P\psi^{\sharp}$ where $\psi$ is a representative of the class $[\psi+df]\in\gu_0^*$ and $\sharp$ is the inverse to $\flat$.

The pairing between $[\xi^{\flat}]\in\gu_0^*$ and $\eta\in\gu_0$ is
\begin{equation}
\label{e:dual-p}
     \vv<[\xi^{\flat}],\eta> 
 = \ww<\xi,\eta>
 = \int_M \vv<\xi(x),\eta(x)>\,dx 
\end{equation} 
which is well-defined independently of the representative of  $[\xi^{\flat}]$.

We define the bracket $[.,.]$ to be the negative of the usual Lie bracket:
$
 [\xi,\eta] := -\nabla_{\xi}\eta + \nabla_{\eta}\xi
$
where $\nabla_{\xi}\eta=\vv<\xi,\nabla>\eta$. This choice of sign is compatible with \cite{A66,AK98,Michor06}. 

Define the operators $\ad$ and $\ad^*$ by  $\ad(\xi).\eta = [\xi,\eta]$ and $\ad^*(\xi).[\eta^{\flat}] = [-\eta^{\flat}\circ\ad(\xi)]$ for $\xi,\eta\in\gu_0$. 
Further, define $\ad^{\top}$ by $\ad^{\top}(\xi) := -\ad(\xi)^{\top}$ which is the transpose with respect to the weak inner product $\ww<.,.>$:
\[
 \ww<\ad(\xi)^{\top}.\eta, \zeta> 
 = \ww<\eta,\ad(\xi).\zeta>.
\]
Sobolev spaces are not closed under the Lie bracket, whence the operations $\ad$, $\ad^{\top}$ and $\ad^*$ lose one Sobolev index, that is $\ad(\xi).\eta=[\xi,\eta]\in\gu_0^{s-1}$ for $\xi,\eta\in\gu_0=\gu_0^s$.
Thus $\gu_0$ is not a Lie algebra.

\begin{lemma}\label{lem:I}
We have 
$\ad^*(\xi).\check{\mu}(\eta) = -\check{\mu}(\nabla_{\xi}\eta + \xi'\otimes\eta)$
and $\ad^{\top}(\xi).\eta = -P(\nabla_{\xi}\eta + \xi'\otimes\eta)$.
\end{lemma}

\begin{proof}
Let $\xi,\eta,\zeta\in\gu_0$
\begin{align*}
    (\ad^*(\xi).\check{\mu}(\eta))(\zeta)
    &= -\int_M\vv<\eta , [\xi,\zeta]>\,dx
     = \int_M \vv<\eta, \nabla_{\xi}\zeta - \nabla_{\zeta}\xi>\, dx \\
    &= \int_M\Big(\d(\vv<\eta,\zeta>\xi)
                  - \vv<\nabla_{\xi} \eta , \zeta>
                  - \vv<\xi'\otimes \eta , \zeta> \Big)\,dx
\end{align*}
\end{proof}

\subsection{Symplectic structure on $G_0\times\gu_0$}
Due to the identification $T^*G_0\cong G_0\times\gu_0^*\cong G_0\times\gu_0$, the space $G_0\times\gu_0$ carries a natural symplectic structure. We follow, and slightly adapt, the exposition of \cite[Section~4]{Michor06} to describe this structure. 
The second tangent bundle is expressed, again via right multiplication, as
$
 T(G_0\times\gu_0)
 = G_0\times\gu_0\times\gu_0\times\gu_0
 = \set{(g,\xi,\dot{g},\dot{\xi})} 
$.
The natural exact symplectic form $\Om$ on $G_0\times\gu_0$ is given by 
\begin{equation}
    \label{e:sp_form}
    \Om_{(g,\xi)}\Big((\dot{g}_1,\dot{\xi}_1),(\dot{g}_2,\dot{\xi}_2)\Big)
    =
    \ww<\dot{\xi}_2,\dot{g}_1> + 
    \ww<-\dot{\xi}_1 + \ad^{\top}(\dot{g}_1).\xi, \dot{g}_2>.   
\end{equation}
The Hamiltonian vector field $X_f$ of a function $f$ is defined by $i(X_f)\Om = df$. It is given explicitly by
\begin{equation}
    \label{e:ham_vf}
    X_f(g,\xi)
    = \Big(
    \textup{grad}_2 f\,(g,\xi),
    \ad^{\top}(\textup{grad}_2 f\,(g,\xi)).\xi - \textup{grad}_1 f\,(g,\xi)
    \Big)
\end{equation}
where $\textup{grad}_1 f$ and $\textup{grad}_2 f$ are the partial gradients. 
Since $\mu$ is only a weak Riemannian metric, the gradients need not exist. But if they do, then so does the Hamiltonian vector field, and it is given by the above formula. 

\subsection{The momentum map}
The particle relabeling symmetry group $G_0$ acts on itself by  right multiplication $R^g: G^s_0\to G^s_0$. The tangent lifted action acts on $G^s_0\times\gu^s_0$ via $TR^g: (k,\xi)\mapsto(kg,\xi)$. 
The associated momentum map is given by 
\begin{equation}
    \label{e:momap}
 J: G^s_0\times\gu^s_0\longto(\gu^{s-1}_0)^*,
 \;
 (g,\xi)\longmapsto \Ad( g)^*[\xi^\flat] = [g^*\xi^{\flat}]
\end{equation} 
where $\Ad( g)\xi = TL_{ g}.(TR^{ g})^{-1}.\xi = T g\circ\xi\circ g^{-1}$, $L_{ g}$ is left multiplication by $g$, and $\Ad(g)^*=:\Ad^*(g^{-1})$ is the adjoint of $\Ad(g)$ with respect to the dual pairing. Further,  $g^*\xi^{\flat}$ is the pullback of the one-form $\xi^{\flat}$ by the diffeomorphism $g$. Since the pullback involves differentiation, the result is only of Sobolev class $H^{s-1}$. Thus $J$ is not a momentum map in the classical sense. Nevertheless, we can use it as long as  the expression $J(g,\xi)$  makes sense and is of sufficient regularity. 


\begin{proposition}[Kelvin Circulation Theorem]
\label{prop:kct_gen}
Let $\mathcal{N}\subset\mathbb{N}$ and consider Hamiltonian functions $h, f^a: G_0^s\times\gu_0^s\to\R$ for $a\in\mathcal{N}$. Let $(W^a)$ for $a\in\mathcal{N}$ be a sequence of pairwise independent Brownian motions. Assume $\gamma_t = (g_t,\xi_t)$ solves the Stratonovich Hamiltonian SDE
\[
\delta\gamma_t = X_h(\gamma_t)\,\delta t + \sum_{a\in\mathcal{N}}X_{f^a}(\gamma_t)\;\delta W^a_t
\]
for $t\in[0,T]$. If $h$ and $f^a$ are invariant under $G^s_0$, then, for any closed smooth curve $C$ in $M$,
\[
 \delta\int_{(g_t)_*C}\xi_t^{\flat} = 0.
\]
\end{proposition}

\begin{proof}
By the stochastic Noether Theorem (\cite{LCO08}),  $J(g_t,\xi_t)$ is constant in $t$. Thus
\[
 \delta (\Ad( g_t)^*.[\xi_t^{\flat}])
 = 0 
 \in\gu_0^*,
\]
whence there is a time dependent function $p_t$ such that $\delta(\xi_t\circ\Ad( g_t)) = d p_t$. \revise{Full} differentials vanish when integrated over a closed loop $C$. Therefore,
\[
  0 
  = \int_C \delta( \xi_t^{\flat}\circ\Ad( g_t) )
  = \delta \int_C g_t^*\xi^{\flat} 
  = \delta \int_{( g_t)_*C}\xi_t^{\flat}.
\]
\end{proof}

The \emph{mechanical connection} is defined by $A := \check{\mu}^{-1}\circ J$, that is
\begin{equation}
    \label{e:mech_con}
 A: G^s_0\times\gu^s_0\longto \gu^{s-1}_0,
 \;
 (g,\xi)\longmapsto 
 \check{\mu}^{-1}\Big(\revise{\Ad^*(g^{-1})}[\xi^\flat]\Big)
 =: \Ad^{\top}(g^{-1})(\xi) .  
\end{equation}

\section{A Hamiltonian interacting particle system for fluid dynamics}\label{sec:ham-ips}

\subsection{Translational kinetic energy of a particle ensemble}\label{sec:det_dyn}
Consider a system of $N$ identical particles in $\R$ with positions $q_1,\dots,q_N\in\R$ and total mass $M=1$. 
Suppose that the dynamics of the system is determined entirely by the translational kinetic energy of the particle ensemble.
The phase space of the system is $T^*\R^N$ and we identify $T^*\R^N=T\R^N$ via the standard inner product $\vv<.,.>$. Let $\wt{\om} = \by{1}{N}\om$ denote the natural symplectic form on $T\R$. Then the product symplectic form 
is 
\[
 \wt{\om^N} 
 := \by{1}{N}\om^N 
 := \sum_{j=1}^N\pi_j^*\wt{\om}
 = \by{1}{N}\sum_{j=1}^N dq^j\wedge dv^j  
\]
where $(q,v) = (q^1,\dots,q^N,v^1,\dots,v^N)$ are coordinates on $T\R^N$ and 
$\pi_j: T\R^N\to T\R$ is the projection onto the $j$-th factor. 
Since $\by{1}{N}\sum_{j=1}^N v^j$ is the translational momentum of the system, the kinetic energy is given by 
\[
 \wt{\mathcal{H}}(q,v)
 := \by{1}{N}\mathcal{H}(q,v)
 := \by{1}{2N^2}\sum_{j,k=1}^N\vv<v^j,v^k>
 = \by{1}{2N^2}\sum_{j=1}^N\Big(\vv<v^j,v^j> + \sum_{k\neq j}\vv<v^j,v^k>\Big). 
\]
Since $i(X)\wt{\om^N} = d\wt{\mathcal{H}}$ if and only if $i(X)\om^N = d\mathcal{H}$, the dynamics of the system is given by the Hamiltonian vector field $X=X_{\mathcal{H}}$ with respect to $\om^N$. 
Denote the $i$-th component of the Hamiltonian vector field by $X^{i}_{\mathcal{H}}(q,v) = T\pi_i.X_{\mathcal{H}}(q,v)$. We remark that $X^{i}_{\mathcal{H}}(q,v)$ is not a vector field on $T\R$. Since $\om^N = \sum_{j=1}^N\pi_j^*\om$, it follows that the $i$-th component $\gamma^i(t)$ of the integral curve $\gamma(t) = (q(t),v(t))$ is determined by
\begin{equation}
    \label{e:det_dyn}
    \dd{t}{}\gamma^i(t)
    = X_{\mathcal{H}}^{i}(\gamma(t))
    = 
    \left(
    \begin{matrix}
    \dd{v^i}{}\mathcal{H}(\gamma(t)) \\
    -\dd{q^i}{}\mathcal{H}(\gamma(t)) 
    \end{matrix}
    \right)
    = 
    \left(
    \begin{matrix}
    \by{1}{N}\sum_{j=1}^N v^j(t) \\
     0
    \end{matrix}
    \right)
\end{equation}

\subsection{The Hamiltonian construction of Brownian motion in finite dimensions}
Let $(Q,\mu^Q)$ be a Riemannian manifold with dimension $\dim Q=m$. For a vector field $X\in\X(Q)$ we define the \emph{momentum function} $F_X: TQ\to\R$ by 
\begin{equation}
    \label{e:mom_fun}
    F_X(Y) = \mu^Q(X(\tau(Y)), Y)
\end{equation}
where $\tau: TQ\to Q$ is the tangent projection. 
We use the metric isomorphsim $\check{\mu}^Q:TQ\to T^*Q$ to transfer the natural symplectic form $\om^{T^*Q}$ on $T^*Q$ to a symplectic form $\om^{TQ} = (\check{\mu}^Q)^*\om^{T^*Q}$ on $TQ$. Let $X_f$ denote the Hamiltonian vector field on $TQ$ with respect to $\om^{TQ}$. 

Assume now that there is a global orthonormal frame $(x_{\alpha})_{\alpha=1}^m$ on $Q$ such that $\nabla_{x_{\alpha}}x_{\alpha}=0$. We shall abbreviate $F^{\alpha}:=F_{x_{\alpha}}$. 
Consider Brownian motion $(W^{\alpha})_{\alpha=1}^m$ in $\R^m$ and assume that the stochastic process $\gamma$ in $TQ$ is a solution to the Stratonovich SDE
\begin{equation}
    \label{e:ham_bm}
    \delta\gamma
    = \sum X_{F^{\alpha}}(\gamma)\,\delta W^{\alpha}.
\end{equation}
It is shown in \cite[Section~3.4]{LCO08} that then $\tau\circ\gamma$ is Brownian motion in $Q$. 

\subsection{Stochastic perturbation of \eqref{e:det_dyn}}
Let us consider the case $Q=\R$ and $\om^{T\mathbb{R}} = \om = dq\wedge dv$ where $(q,v)$ are coordinates on $T\R$. 
Let $\wt{F}: T\R\to\R$ $(q,v)\mapsto \by{1}{N}e v =: \by{1}{N}F(q,v)$ where $e=1$ is viewed as the standard basis vector in $\R$. Notice, as in Section~\ref{sec:det_dyn}, that $i(X)\wt{\om} = d\wt{F}$ if and only if $i(X)\om = dF$. Let $X_F$ be the Hamiltonian vector field of $F$ with respect to $\om$. The stochastic perturbation of Equation~\eqref{e:det_dyn} is therefore given by
\begin{equation}
    \label{e:stoch_dyn}
    \delta \gamma^i_t
    = X_{\mathcal{H}}^i(\gamma_t)\,\delta t
     + X_{F}(\gamma_t^{\revise{i}})\,\delta W_t^i
\end{equation}
where $W_t^1,\dots,W_t^N$ are independent copies of Brownian motion in $\R$ and $\gamma_t=(\gamma^i_t)_i$. 
This describes a particle in an ensemble, where the energy of each (identical) particle is given as a sum of two components:
\begin{enumerate}[\up (1)]
    \item a deterministic part due to the translational momentum of the ensemble;
    \item an internal energy that is modeled as a Brownian motion. 
\end{enumerate}
The corresponding model for the full system of particles is given by the stochastic Hamiltonian equation
\begin{equation}
    \label{e:stoch_dyn2}
    \delta \gamma_t
    = X_{\mathcal{H}}(\gamma(t))\,\delta t
     + \sum_{i=1}^N X_{\pi_i^*F}(\gamma_t)\,\delta W_t^i
\end{equation}
on $T\R^N$, where $X_{\pi_i^*F}$ is the Hamiltonian vector field of the pullback $\pi_i^*F$  to $T\R^N$.
We have $T\pi_i.X_{\pi_i^*F} = X_F\circ\pi_i$. 
Equation \eqref{e:stoch_dyn2} is the motivation for \eqref{e:ham_ips1} below.

\subsection{An orthogonal system on $\gu_0$}
Let
\[
 \mathbb{Z}_n^+
 := \set{k\in\mathbb{Z}_n: k_1>0
 \textup{ or, for } i=2,\ldots,n,
 k_1=\ldots=k_{i-1}=0, k_i>0
 }.
\]
For $k\in\mathbb{Z}^+_n$ let $k_1^{\bot},\ldots,k_{n-1}^{\bot}$ denote a choice of pairwise orthogonal vectors in $\R^n$ such  that $|k_i^{\bot}|=|k|$ and $\vv<k_i^{\bot},k>=0$ for all $i=1,\ldots,n-1$. 

In \cite[Appendix]{CS09} the following system of vectors is introduced:
\begin{equation*}
    A_{(k,i)} =  \frac{1}{|k|^{s+1}}\cos\vv<k,x>k^{\bot}_i,\;
    B_{(k,i)} = \frac{1}{|k|^{s+1}}\sin\vv<k,x>k^{\bot}_i,\;
    A_{(0,j)} = e_j 
\end{equation*}
Here $e_j\in\R^n$ is the $j$th standard vector. 
By slight abuse of notation we identify these vectors with their corresponding right invariant vector fields on $G_0$. 

Further, we shall make use of the multi-index notation $\alpha = (k,i,a)$ where $k\in\mathbb{Z}_n^+$ and $a=0,1,2$ such that
\begin{align*}
    X_{\alpha} &= A_{(0,i)}
    \textup{ with } i=1,\ldots,n
    \textup{ if } a=0\\
    X_{\alpha} &= A_{(k,i)} 
        \textup{ with } i=1,\ldots,n-1
   \textup{ if } a=1\\
    X_{\alpha} &= B_{(k,i)} 
        \textup{ with } i=1,\ldots,n-1
    \textup{ if } a=2  
\end{align*}
Thus by a sum over $\alpha$ we shall mean a sum over these multi-indices, and this notation for $X_{\alpha}$ will be used throughout the rest of the paper. 

It is shown in \cite[Appendix]{CS09} that the $X_{\alpha}$ form an orthogonal system of basis vectors in $\gu_0$ such that 
\begin{equation}\label{e:nablaXX} 
 \nabla_{X_{\alpha}}X_{\alpha} = \vv<X_{\alpha},\nabla>X_{\alpha} = 0.
\end{equation} 
and, for $\xi\in\X(M)$, 
\begin{equation}\label{e:Delta}
 \sum \nabla_{X_{\alpha}}\nabla_{X_{\alpha}}\xi = c^s\Delta\xi
\end{equation}
where $c^s>0$ is a constant and $\Delta$ is the vector Laplacian.

\begin{proposition}[\cite{CC07,C,DZ}]\label{prop:bm}
Let $W_t = \sum X_{\alpha} W_t^{\alpha}$, where $W_t^{\alpha}$ are independent copies of Brownian motion in $\R$. Then $W$ defines (a version of) Brownian motion (i.e., cylindrical Wiener process) in $\gu_0$. 
\end{proposition}

\subsection{The Hamiltonian Interacting Particle System (IPS)}
We consider the phase space $(G_0\times\gu_0)^N$ for (large) $N\in\mathbb{N}$.
Let $\pi_i: (G_0\times\gu_0)^N\to G_0\times\gu_0$ be the projection onto the $i$-th factor. 
Let $\Om$, from now on, be the symplectic form \eqref{e:sp_form} on $G_0\times\gu_0$ and $\Om^N := \sum_{i=1}^N\pi_i^*\Om$ the product symplectic form. 

Following \eqref{e:stoch_dyn2} we define the vector valued Hamiltonian $(\mathcal{H}^N,F_i^{\alpha})_{i,\alpha}$ by  
\begin{align}
    \mathcal{H}^N(\Gamma) &= \by{1}{2N}\sum_{i,j=1}^N\ww<\xi^i,\xi^j> \\  
    F_i^{\alpha}(\Gamma)  &= \ww<\xi^i, X_{\alpha}> 
\end{align} 
where  $\Gamma = (\Gamma_i)_{i=1}^N = (g^i,\xi^i)_{i=1}^N\in(G_0\times\gu_0)^N$.
For each $\alpha$ let $(W^{\alpha,i})_{i=1}^N$ be Brownian motion in $\R^N$ such that $W^{\alpha,i}$ and $W^{\beta,j}$ are independent for $(\alpha,i)\neq(\beta,j)$.  For $\eta>0$ we use $\nu=\sqrt{\by{2\eta}{c^s}}$ to scale the white noise term. 
The resulting Stratonovich SDE in $(G_0\times\gu_0)^N$ associated to the collection $(\mathcal{H}^N,F_i^{\alpha})_{i,\alpha}$ is  
\begin{equation}
   \label{e:ham_ips1}
    \delta \Gamma_t
    = 
    X_{\mathcal{H}^N}(\Gamma_t)\,\delta t 
    + \nu\sum_{\revise{i},\alpha}X_{\pi_i^*F_i^{\alpha}}(\Gamma_t)\,\delta W^{\alpha,i}_t
\end{equation}
Note that $T\pi_i.X_{\pi_i^*F_i^{\alpha}}(\Gamma) = X_{F_i^{\alpha}}(\pi_i(\Gamma))$ where $X_{\pi_i^*F_i^{\alpha}}$ is the Hamiltonian vector field on $(G_0\times\gu_0)^N$. 
The term $\sum_{\revise{i},\alpha}X_{\pi_i^* F_i^{\alpha}}(\Gamma_t)\,\delta W^{\alpha,i}_t$ is an infinite dimensional version of Equation~\eqref{e:ham_bm}. This construction is justified by Proposition~\ref{prop:bm}.
Equation~\eqref{e:ham_vf} yields
\begin{align*}
    T\pi_i . X_{\mathcal{H}^N}(\Gamma) 
    &= 
    \Big(
    \by{1}{N}\sum_{j=1}^N\xi^j, \ad^{\top}(\by{1}{N}\sum_{j=1}^N\xi^j).\xi^i \Big)\\
    T\pi_i . X_{F_i^{\alpha}}(\Gamma)
    &=
    \Big( X_{\alpha}, \ad^{\top}(X_{\alpha}).\xi^i \Big) 
\end{align*}
Therefore, the Hamiltonian equations of motion \eqref{e:ham_ips1} for $\Gamma_t = (g^i_t,\xi^i_t)_{i=1}^N$ can be written as
\begin{align}
    \label{e:ham_ips2}
    \delta g^i_t
    &= TR^{g^i_t}\Big(
        \by{1}{N}\sum_{j=1}^N\xi^j_t \,\delta t
        + \nu\sum_{\alpha} X_{\alpha}\,\delta W^{\alpha,i} \Big)\\
    \label{e:ham_ips3}
    \delta \xi^i_t
    &= \ad^{\top}(\by{1}{N}\sum_{j=1}^N\xi^j_t).\xi^i_t\,\delta t
    + \nu\sum_{\alpha}\ad^{\top}(X_{\alpha}).\xi^i_t\,\delta W^{\alpha,i} 
\end{align}


\section{The mean field system}\label{sec:main}

\subsection{The mean field limit}
Stochastic mean field equations in a Hilbert space setting are treated in \cite{AD95}. See also \cite{Sni91,Mel96} for finite dimensional convergence results. Consider equations~\eqref{e:ham_ips2} and \eqref{e:ham_ips3} together with a deterministic initial condition $\Gamma^i_0=(e,u_0)$ for $i=1,\dots,N$ and where $e\in G_0$ is the identity in the group.
For the following we assume there exists a $T>0$, independent of $N$ but with a possible dependence on the initial condition, such that the system~\eqref{e:ham_ips2} and \eqref{e:ham_ips3} has a strong solution on $[0,T]$ in the usual SDE sense. Further, we assume that the mean field limit of \eqref{e:ham_ips2} and \eqref{e:ham_ips3} exists in the sense of \cite{AD95}.
This means:
\begin{enumerate}[\up (1)]
    \item 
    The empirical mean $\by{1}{N}\sum_{j=1}^N\xi^j$ converges in probability to a smooth curve $u: [0,T]\to\gu_0$ as $N\to\infty$. That is, $u=u(t,x)=u_t(x)$ is a time dependent vector field.
    \item
    The stochastic process $(g^1_t,\xi^1_t)$ converges in the appropriate norm, for $N\to\infty$, to a stochastic process $(g_t,\xi_t)$, and the limiting process solves 
   \begin{align}
    \label{e:ham_mf1}
    \delta g_t
    &= TR^{g_t}\Big(
        u_t \,\delta t
        + \nu\sum_{\alpha} X_{\alpha}\,\delta W^{\alpha} \Big)\\
    \label{e:ham_mf2}
    \delta \xi_t
    &= \ad^{\top}(u_t).\xi_t\,\delta t
    + \nu\sum_{\alpha}\ad^{\top}(X_{\alpha}).\xi_t\,\delta W^{\alpha} 
\end{align}
where $W^{\alpha} = W^{\alpha,1}$ is a sequence of independent Brownian motions. Since the particles are identical there is no loss in generality in considering the limit of $(g^1(t),\xi^1(t))$.
\item
$u_t = E[\xi_t]$.
\end{enumerate}

Equations~\eqref{e:ham_mf1} and \eqref{e:ham_mf2} are mean field SDEs and their solutions are to be understood in the sense of Definition~\ref{def:mf_sol}. 

\begin{definition}[Strong solution of a mean field equation]\label{def:mf_sol}
Let $V$ be a Hilbert space and $W$ a cylindrical Wiener process in $V$. 
A $V$-valued stochastic process $\xi_t$ with $t\in[0,T]$ is a strong solution to the mean-field or McKean-Vlasov Stratonovich SDE 
\[
 \delta \xi_t = f(\xi_t,\mu_t)\,\delta t + g(\xi_t)\,\delta W_t
\] 
if
\begin{enumerate}[\up (1)]
\item
$\mu_t$ is the law of $\xi_t$, i.e., $\mu_t = (\xi_t)_*P$,
\item
$\xi_t$ is adapted to $(\F_t)$,
\item
$t\mapsto\xi_t$ is continuous $P$-a.s.,
\item
$\xi_t 
 = \xi_0 + \int_0^t f(\xi_s,\mu_s) \,ds 
   + \int_0^t g(\xi_s)\,\delta W_s $ 
for all $t\in[0,T]$ $P$-a.s.
\end{enumerate}
\end{definition}

Notice that, once the prescription $t\mapsto\mu_t$ of the law of $\xi_t$ is found, the concept of a mean field Stratonovich SDE is not different from that of a time dependent Stratonovich SDE. The Stratonovich integral above is hence to be understood as 
\[
 \int_0^t g(\xi_s)\,\delta W_s
 = \int_0^t g(\xi_s)\, d W_s
 + \by{1}{2}[g(\xi),W]_t.
\]
See e.g.\ \cite[Page~82]{Pro}.

Equations~\eqref{e:ham_mf1} and \eqref{e:ham_mf2} are defined on the phase space $G_0\times\gu_0$. They are Hamiltonian with respect to the symplectic structure \eqref{e:sp_form}: Consider the time dependent Hamiltonian
$H^t(g,\xi) = \ww<\xi,u_t>$ and the associated time dependent Hamiltonian vector field $X_{H^t}$ given by $i(X_{H^t})\Om = dH^t$. Equations~\eqref{e:ham_mf1} and \eqref{e:ham_mf2} are thus equivalent to 
\begin{equation}
    \label{e:ham_mf3}
  \delta\Gamma_t = X_{H^t}(\Gamma_t)\,\delta t
  + \sum_{\alpha}X_{F^{\alpha}}(\Gamma_t)\,\delta W^{\alpha}_t 
\end{equation}
where $\Gamma_t = (g_t,\xi_t)$ and $F^{\alpha}(g,\xi) = \ww<\xi,X_{\alpha}>$. 

By right invariance of the Hamiltonian system~\eqref{e:ham_mf1} and \eqref{e:ham_mf2}, we can pass to $\gu_0$, which is the Lie-Poisson reduction of $G_0\times\gu_0$.
Therefore, using the expression for $\ad^{\top}$ in Lemma~\ref{lem:I}, we consider, for $\eta>0$ and $\nu=\sqrt{\by{2\eta}{c^s}}$,  
the mean-field Stratonovich equation in $\gu_0$
\begin{align}\label{e:sdeP}
 \delta \xi_t 
 &= - P\Big(\nabla_{u_t}\xi_t + u_t'\otimes\xi_t \Big)\,\delta t
    - \nu \sum P\Big(\nabla_{X_{\alpha}}\xi_t + X_{\alpha}'\otimes\xi_t     
    \Big) \,\delta W_t^{\alpha}\\
 u_t &= E[\xi_t] \notag \\
  \xi_0 &= u_0  \notag 
\end{align}
where $u_0\in\gu_0$ is the initial condition. 
In Proposition~\ref{prop:KC}, Equation~\eqref{e:ham_mf1} will be used as a reconstruction equation to pass again to the full phase space $G_0\times\gu_0$. 
\revise{
\begin{remark}
The stochastic velocity along a path of the fluid motion is given, in the Euler picture, by Equation~\eqref{e:ham_mf1} as $(TR^{g_t})^{-1}\delta g_t$. Since $(TR^{g_t})^{-1}\delta g_t$ and $\xi_t$ do not coincide, we refer to the latter as \emph{stochastic specific momentum in velocity space} or, more briefly, as \emph{stochastic momentum}. This terminology is justified, because $\check{\mu}(\xi_t) = \rho\check{\mu}(\xi_t)\in\gu_0^*$ is the momentum density, where $\check{\mu}$ is the inertia tensor~\eqref{equ:I} and $\rho=1$ is the fluid density. See also \cite[Remark~5]{CFH17}.
The distinction between velocity and momentum in velocity space is necessary, because the vector valued Hamiltonian, $(H^t,F^{\alpha})$ in \eqref{e:ham_mf3}, is not the kinetic energy Hamiltonian associated to $\mu$.  
\end{remark}
}

\subsection{Hamiltonian mean field system for the Navier-Stokes equation}
For $x\in M$ we define the (point) evaluation map by  $\ev_x: \gu\to T_x M = \R^n$, $\xi\mapsto\xi(x)$. 

\begin{lemma}\label{lem:1}
Consider $\hat{X_{\alpha}}: \gu^s\to\gu^{s-1}$, $\xi\mapsto \nabla_{X_{\alpha}}\xi + X_{\alpha}'\otimes\xi$
and 
$\hat{Y_{\alpha}}: \gu^s\to\gu_0^{s-1}$, $\xi\mapsto P(\nabla_{X_{\alpha}}\xi + X_{\alpha}'\otimes\xi)$. Then
\[
 \sum\hat{X_{\alpha}}\hat{X_{\alpha}}(\ev_x)(\xi) 
 = c^s\Delta\xi(x)
\textup{ and }
 \sum\hat{Y_{\alpha}}\hat{Y_{\alpha}}(\ev_x)(\xi) 
 = c^s\Delta P \xi(x) .
\]
\end{lemma}

Note that  $\hat{Y_{\alpha}}(\xi) = -\ad^{\top}(X_{\alpha}).\xi$ for $\xi\in\gu_0$. 

\begin{proof}
The directional derivative of $\ev_x$ along $\hat{X_{\alpha}}$ is not necessarily well-defined, since the latter is not a proper vector field on $\gu^s$ but takes values in $\gu^{s-1}$. However, since $\ev_x$ is linear, the G\^ateaux derivative is
\[
 \hat{X_{\alpha}}(\ev_x)(\xi) 
 = d\,\ev_x(\hat{X_{\alpha}})(\xi)
 = \ddo\ev_x(\xi + t\hat{X_{\alpha}}(\xi))
 = \hat{X_{\alpha}}(\xi)(x)
\]
which clearly exists.
Similarly,
\begin{align*}
    \sum\hat{X_{\alpha}}\hat{X_{\alpha}}(\ev_x)(\xi)
    &= \ev_x\sum\Big(
        \nabla_{X_{\alpha}}\nabla_{X_{\alpha}}\xi
        + \nabla_{X_{\alpha}}(X_{\alpha}'\otimes\xi)
        + X_{\alpha}'\otimes\nabla_{X_{\alpha}}\xi
        + X_{\alpha}'\otimes X_{\alpha}'\otimes\xi 
        \Big)
\end{align*}
By \eqref{e:Delta} the first term becomes $\sum\nabla_{X_{\alpha}}\nabla_{X_{\alpha}}\xi = c^s\Delta\xi$.
The other terms cancel.
Indeed, for $\alpha=(k,i,0)$ this is clear.
Note that 
\begin{align}
\label{e:A'}
    A_{(k,i)}' &= -B_{(k,i)}\otimes k \textup{ and }
    B_{(k,i)}' = A_{(k,i)}\otimes k. 
\end{align}
Let us consider $\alpha=(k,i,1)$, the case $\alpha=(k,i,2)$ being analogous.
Now,
\[
 \nabla_{A_{(k,i)}}(A^{\revise{'}}_{(k,i)}\otimes\xi)
 = \nabla_{A_{(k,i)}}(-\vv<k,\xi>B_{(k,i)})
 = -\vv<k,\xi>\nabla_{A_{(k,i)}}B_{(k,i)} 
   - \vv<k,\nabla_{A_{(k,i)}}\xi> B_{(k,i)}
\]
where the first term vanishes because of \eqref{e:nablaXX} and the second term cancels in the sum because of the asymmetric sign change in \eqref{e:A'}.
We have thus shown that
\begin{equation}
    \label{e:alpha_sum}
    \sum \nabla_{X_{\alpha}}(X_{\alpha}'\otimes\xi) = 0.
\end{equation}
The terms 
\[
 A_{(k,i)}'\otimes\nabla_{A_{(k,i)}}\xi
 \textup{ and }
 B_{(k,i)}'\otimes\nabla_{B_{(k,i)}}\xi
 = -A_{(k,i)}'\otimes\nabla_{A_{(k,i)}}\xi 
\]
also cancel in the sum. Finally,
\begin{align*}
    A_{(k,i)}'\otimes A_{(k,i)}'\otimes\xi 
    &= A_{(k,i)}'\otimes(-B_{(k,i)}\otimes k\otimes\xi)
    = \frac{1}{|k|^{2s+2}}\sin^2\vv<k,x> k_i^{\bot} k^t k_i^{\bot} k^t \xi
    = 0
\end{align*}
since $k^t k_i^{\bot} = 0$ by construction,
where $k^t$ is the transpose of $k$ . 
For the second statement we note that 
\begin{align*}
    &\sum\hat{Y_{\alpha}}\hat{Y_{\alpha}}(\ev_x)
    = \ev_x\circ P \Big(\hat{X_{\alpha}}(\hat{X_{\alpha}}(\xi)) 
        - \hat{X_{\alpha}}(\nabla f_{\alpha}^{\xi})
        \Big)\\
    &= \ev_x\circ P\Big(
        c^s\Delta\xi
        - \nabla_{X_{\alpha}}(\nabla f_{\alpha}^{\xi})
        - X_{\alpha}'\otimes\nabla f_{\alpha}^{\xi}  
        \Big)
    = \ev_x\circ P\Big(
        c^s\Delta\xi
        - \nabla \nabla_{X_{\alpha}}f_{\alpha}^{\xi}
        \Big)
    = c^s\Delta P\xi (x)
\end{align*}
where $f_{\alpha}^{\xi} = \Delta^{-1}\d(\hat{X_{\alpha}}(\xi))$.
\end{proof}

\begin{theorem}\label{thm:main}
\begin{enumerate}[\up (1)]
    \item 
    If $\xi_t\in\gu_0$ is a strong solution to \eqref{e:sdeP} on $[0,T]$ such that $\xi_0=u_0\in\gu_0$, then $u(t,x) = E[\ev_x(\xi_t)]$ satisfies the Navier-Stokes equation with $u(0,x)=u_0(x)$ for incompressible flow in $M$ on $[0,T]$:
\begin{align}\label{e:nse}
    \dd{t}{} u &= -P\nabla_u u + \eta\Delta u \textup{ and }
    \d\,u = 0
\end{align}
\item
Conversely,
suppose  $u$ is a smooth solution of \eqref{e:nse} on $[0,T]$ with initial condition $u_0$, and $g^u$ is a strong solution of \eqref{e:ham_mf2} on $[0,T]$ with $g^u_0=e$. Let $\zeta_t := Ad^{\top}(g^u_t).u_0$ where $\Ad^{\top}$ is defined in \eqref{e:mech_con}.
It follows that $u=E[\zeta]$, and $(g_t^u, \zeta_t)$ is a solution of \eqref{e:ham_mf3}.
In particular, $\zeta_t$ solves the mean field system \eqref{e:sdeP}. 
\end{enumerate}
\end{theorem} 

In Part (1), Equation~\eqref{e:nse} is thus the McKean-Vlasov equation for $\ev_x: \gu_{\revise{0}}\to\R^n$ corresponding to the non-linear generator of $\xi_t$.
In Part (2),
we regard $g^u_t$ as a \emph{Gaussian perturbation} of the Lagrangian trajectory defined by $u$. This is justified by Proposition~\ref{prop:bm}.  
 
\begin{proof} 
Ad (1).
Note that $\d\,u_t = E[\d\,\xi_t] = 0$.
Let us write the Ito version of \eqref{e:sdeP} as 
$d\xi_t = b(u_t,\xi_t)\,dt - \nu\sum\hat{Y_{\alpha}}(\xi_t)\,dW_t^{\alpha}$ where we use the notation from Lemma~\ref{lem:1}.
By linearity of $\hat{Y_{\alpha}}$ and Lemma~\ref{lem:1}, we have for the quadratic variation
\begin{align*}
\sum\Big[-\nu\hat{Y_{\alpha}}(\xi_.), W^{\alpha}_. \Big]_t 
&= -\nu\sum\Big[
    \int_0^{\cdot}\hat{Y_{\alpha}}(b(u_s,\xi_s))\, ds 
    -
    \nu\int_0^{\cdot}\hat{Y_{\alpha}}(\hat{Y_{\alpha}}(\xi_s))\,dW_s^{\beta},
    W^{\alpha}_.
 \Big]_t \\
& = 
\nu^2\sum\delta_{\alpha\beta}\int_0^t \hat{Y_{\alpha}}(\hat{Y_{\alpha}}(\xi_s)) \;ds 
= 
c^s\nu^2\int_0^t \Delta\xi_s\, ds
\end{align*}
where $\delta_{\alpha\beta}$ is the Kronecker delta.
(See \cite{Pro} or \cite[Section~3.4]{DZ}.) 
Therefore, the Ito version of \eqref{e:sdeP} is
\begin{equation}
\label{e:sdeP-Ito}
d \xi_t 
 = \Big(-P(\nabla_{u_t}\xi_t + u_t'\otimes\xi_t)
         + c^s\by{\nu^2}{2}\Delta\xi_t \Big)\,dt
    - \nu \sum P\Big(\nabla_{X_{\alpha}}\xi_t + X_{\alpha}'\otimes\xi\Big) \, dW_t^{\alpha}.
\end{equation}
Since $\eta = c^s\by{\nu^2}{2}$, $P(u'\otimes u) = \by{1}{2}P\nabla\vv<u,u> = 0$, and by linearity of $\xi\mapsto b(u,\xi)$, it follows that
\[
 u_t = u_0 
     + E\Big[\int_0^t\Big(
        - P(\nabla_{u_s}\xi_s 
        + u_s'\otimes\xi_s )
        + \eta\Delta\xi_s 
        \Big)\,ds \Big]
    = u_0 + \int_0^t\Big(-P\nabla_{u_s}u_s + \eta\Delta u_s\Big)\,ds 
\] 
solves the Navier-Stokes equation on $[0,T]$. 

Ad (2). 
Using the notation from \eqref{e:mech_con}, we obtain $A(g_t^u,\zeta_t) = \Ad^{\top}((g^u_t)^{-1}).\Ad(g^u_t).u_0 = u_0$. Therefore, 
\begin{align}
    0 
    &= \delta \ww<\Ad((g_t^u)^{-1}).\zeta_t,\eta> 
    = \ww<\delta \zeta_t, \Ad(g_t^u).\eta> + \ww< \zeta_t, \delta\Ad(g_t^u).\eta> \notag \\
    &=\ww<\delta \zeta_t, \Ad(g_t^u).\eta> 
        + \ww<\zeta_t, \ad(u_t\,\delta t + \sum X_{\alpha}\,\delta W_t^{\alpha}).\Ad(g_t^u).\eta>  \notag \\
    &=\ww<\Ad^{\top}((g^u_t)^{-1}).
        \Big(\delta \zeta_t - \ad^{\top}(u_t\,\delta t + \sum X_{\alpha}\,\delta W_t^{\alpha}).\zeta_t\Big) , \eta >
        \label{e:proof-main}
\end{align}
for arbitrary $\eta\in\gu_0$,
whence $\zeta_t$ satisfies \eqref{e:ham_mf2}.
It remains to show that $E[\zeta_t] = u_t$.
Let $v_t = E[\zeta_t]$ and $w = v-u$. 
Since  \eqref{e:ham_mf2} implies \eqref{e:sdeP-Ito}, it follows that 
\begin{equation}
\label{e:v1}
 \dd{t}{}v 
 =  -P(\nabla_u v + u'\otimes v) + \eta\Delta v.     
\end{equation}
Notice that $\d\,v = \d\,w = 0$. Subtracting the Navier Stokes equation \eqref{e:nse} from \eqref{e:v1} yields \eqref{e:v1} with $v$ replaced by $w$. Therefore, we have the energy identity
\begin{equation}
\label{e:v2}
 \dd{t}{}\by{1}{2}\ww<w,w>
 = \ww<-\nabla_u w - u'\otimes w + \eta\Delta w , w>
 = -\ww<u'\otimes w, w> - \eta\ww<\nabla w, \nabla w> 
\end{equation}
because $\ww<\nabla_u w, w> = 0$ by partial integration. 
This gives 
\[
 \dd{t}{}||w_t||_2^2
 \le
 2|\ww<u_t'\otimes w_t, w_t>|
 \le 
 2\sup_{0\le s\le T}|u'_s|_{\infty}\cdot ||w_t||_2^2
\]
where $|u'_s|_{\infty}$ is the supremum norm and $||\revise{w}||_2^2 = \ww<w,w>$ is the $L^2$-norm. 
The Gronwall Lemma now implies that $||w_t||_2^2 = ||w_0||_2^2 = 0$, whence $v=u$. 
\end{proof}

\begin{corollary}\label{cor:main}
Let $\Phi: C^1([0,T],\gu_0)\to C^1([0,T],\gu_0)$ be defined by $\Phi: u\mapsto E[\Ad^{\top}(g^u).u_0]$. 
Then $u$ is a solution to the Navier-Stokes Equation~\eqref{e:nse} if and only if $u$ is a fixed point of $\Phi$. 
\end{corollary}

\begin{proof}
If $u$ is a solution to \eqref{e:nse} then $u_t = E[\Ad^{\top}(g_t^u).u_0]$ by Theorem~\eqref{thm:main}(2). Conversely, if $\xi_t := \Ad^{\top}(g_t^u).u_0$ and $u_t = E[\xi_t]$ then $\xi_t$ satisfies the mean field Equation~\eqref{e:sdeP}, whence $u$ is a solution of \eqref{e:nse} by Theorem~\eqref{thm:main}(1).
\end{proof}


\subsection{Kelvin Circulation Theorem}

\begin{proposition}\label{prop:KC}
Suppose $\xi_t$ is a solution of \eqref{e:sdeP} with $u_t=E[\xi_t]$.
Define $g_t$ in $G_0$ through the reconstruction Equation~\eqref{e:ham_mf1}, that is 
\begin{equation}
\label{e:del_gamma}
    \delta g_t 
    = \Big(u_t\delta t + \sum X_{\alpha} \delta W_t^{\alpha}\Big)\circ g_t.
\end{equation}
Then, for any closed smooth curve $C$ in $M$, 
\begin{equation}
\label{e:KC}
    \delta \int_{( g_t)_*C}\xi_t^{\flat}
    = 0.
\end{equation}
\end{proposition}

\begin{proof}
Since $H^t$ and $F^{\alpha}$ in \eqref{e:ham_mf3} are invariant under the $G_0$-action, this follows immediately from Proposition~\ref{prop:kct_gen}. However, we can also give an explicit proof: 
\begin{align}
\label{e:kct1}
    \delta \int_{( g_t)_*C}\xi_t^{\flat}
    &= \delta\int_C  g_t^*\xi_t^{\flat}
    = \delta\int_C(\xi_t\circ g_t)^{\flat}\circ T g_t 
    = \delta\int_0^1\vv<\xi_t\circ g_t,
        T g_t.c'(s)>\,ds \\
    &= \int_0^1\Big(\vv<\delta(\xi_t\circ g_t),T g_t.c'(s)>
        + \vv<\xi_t\circ g_t,(Tu_t\,\delta t + \sum TX_{\alpha}\,\delta W_t^{\alpha} )T g_t.c'(s)>
        \Big)\,ds  \notag \\
    &= \int_0^1 
        \vv<\Big(
        \delta\xi_t
        + \nabla_{u_t\,\delta t + \sum X_{\alpha}\,\delta W_t^{\alpha}}\xi_t
        + u_t'\otimes\xi_t \,\delta t
        + \sum X_{\alpha}'\otimes\xi_t \,\delta W_t^{\alpha} 
        \Big)\circ g_t
        , 
        T g_t.c'(s) > \,ds \notag \\
    &= \int_{( g_t)_*C} (\nabla p_t)^{\flat}
    = 0  \notag 
\end{align}
where $c(s)$ is a parametrization of $C$.
\end{proof}

Since 
\begin{equation}
    \label{e:line-stretching}
    u_t'\otimes\xi_t
    \textup{ and }
    X_{\alpha}'\otimes\xi_t
\end{equation}
are obtained by transposing $Tu_t$ and $TX_{\alpha}$, we refer to these as the \emph{line stretching terms}. Compare with \cite[Equ.~(2.37)]{CFH17}.

\section{Energy dissipation}\label{sec:energy}
The energy $\mathcal{E}^d$ associated to a solution $u$ of the Navier-Stokes equation is 
\begin{equation}\label{equ:enedet}
  \mathcal{E}_t^d
  = \by{1}{2}\ww<u_t, u_t>
  = \by{1}{2}\int_M \vv<u(t,x),u(t,x)>\, dx.
\end{equation}
The superscript is to emphasize that this is the energy associated with the deterministic solution $u$.
Since 
$\dd{t}{}\mathcal{E}^d_t 
= -\eta\ww<\nabla u(t,x),\nabla u(t,x)>$, 
it follows that energy dissipates for $\eta>0$. 

For $\xi\in\gu_0$ let
\[
 \mathcal{H}_0(\xi) = \by{1}{2}\ww<\xi,\xi> 
\]
We define the energy of a solution $\xi_t$ of Equation~\eqref{e:sdeP} as 
\begin{equation}\label{e:ensto}
  \mathcal{E}^s_t = E[\mathcal{H}_0(\xi_t)].
\end{equation}
The superscript $s$ stands for stochastic, even though $\mathcal{E}^s_t$ is not a random variable. 

This section is concerned with two questions: 
\begin{enumerate}[\up (1)]
    \item 
    Do conservation or dissipation for $\mathcal{H}_0$ or $\mathcal{E}^s$ hold?
    \item
    Does dissipation of $\mathcal{E}^d$ follow from the stochastic formulation? 
\end{enumerate}

Concerning the first question, note that Equation~\eqref{e:ham_mf2} yields
\begin{align}
\label{e:non-cons}
    \delta\by{1}{2}\ww<\xi_t,\xi_t>
    =
    \ww<\ad^{\top}
        (u_t\,\delta t + \nu\sum X_{\alpha}\,\delta W_t^{\alpha}).\xi_t,
        \xi_t>
    = 
    -\ww<\xi_t, 
    [u_t, \xi_t]\,\delta t 
    + \nu\sum [X_{\alpha}, \xi_t]\,\delta W_t^{\alpha} >
\end{align}
whence commutation of $\xi_t$ with $u_t$ and $X_{\alpha}$ is an obstruction to stochastic energy conservation.
Thus $\mathcal{H}_0$ is not conserved along $\xi_t$. 

Non-conservation of stochastic energy is ubiquitous stochastic geometric mechanics. This is also discussed in \cite{CGD17} from a perspective of multi-scale analysis and, in finite dimensions, in \cite{ACH16}.

\subsection{Non-dissipation}\label{sec:non-diss}
Suppose $\xi_t$ is a solution of Equation~\eqref{e:sdeP-Ito}, which is the Ito formulation of \eqref{e:sdeP}.
Let $\hat{Y}_{\alpha}(\xi)$ as in Lemma~\ref{lem:1}. Then
\begin{equation}\label{e:dH}
    d\ww<\xi_t,\xi_t>
    = 2\ww<\xi_t,d\xi_t> + \nu^2\sum\ww<\hat{Y}_{\alpha}(\xi_t),\hat{Y}_{\alpha}(\xi_t)>\,dt.
\end{equation}
Using $\ww<\xi,\nabla f> = 0$, which follows by partial integration from $\d(\xi)=0$, we note that
\begin{align}
    \ww<\xi_t,d\xi_t>
    &=
    \ww<\xi_t, 
    -(\nabla_{u_t}\xi_t + u_t'\otimes\xi_t - \eta\Delta\xi_t)\,dt
    + \sum\hat{Y}_{\alpha}(\xi_t)\,dW_t^{\alpha} > \\
    &= 
    \ww<\xi_t, 
    (-u_t'\otimes\xi_t + \eta\Delta\xi_t )\,dt
    + \sum\hat{Y}_{\alpha}(\xi_t)\,dW_t^{\alpha} >. 
\end{align}
Note that $\ww<\xi_t, u_t'\otimes\xi_t> = \ww<\xi_t, [u_t,\xi_t]>$. 
Furthermore, 
\begin{align}\label{e:crossterms}
    &\sum\ww<\hat{Y}_{\alpha}(\xi_t),\hat{Y}_{\alpha}(\xi_t)>
    + \sum\ww<\nabla_{X_{\alpha}}\nabla_{X_{\alpha}}\xi , \xi>\\
    &=
    \sum\ww< X_{\alpha}'\otimes\xi , X_{\alpha}'\otimes\xi>
    + 
    \sum\ww<-\nabla f_{\alpha}^{\xi} , \nabla_{X_{\alpha}}\xi_t+X_{\alpha}'\otimes\xi_t >
    \notag 
\end{align}
since the cross terms 
$\sum\ww<\nabla_{X_{\alpha}}\xi, X'_{\alpha}\otimes\xi> 
= -\ww<\xi,\nabla_{X_{\alpha}}(X'_{\alpha}\otimes\xi)> 
= 0$
vanish by Equation~\eqref{e:alpha_sum} and where 
$\Delta f_{\alpha}^{\xi} = \d(\nabla_{X_{\alpha}}\xi_t+X_{\alpha}'\otimes\xi_t)$.
Now,
\begin{equation}
    \sum\ww<-\nabla f_{\alpha}^{\xi} , \nabla_{X_{\alpha}}\xi_t+X_{\alpha}'\otimes\xi_t >
    = 
    - 
    \sum \ww<\nabla f_{\alpha}^{\xi} , \nabla f_{\alpha}^{\xi}>
\end{equation}
and 
\begin{align}\label{e:X_alpha}
    &\sum\ww< X_{\alpha}'\otimes\xi , X_{\alpha}'\otimes\xi>\\
    &=
    \sum_{k\in\mathbb{Z}_n^+}\sum_{i=1}^{n-1}\Big(
        \ww<B_{(k,i)}\otimes k\otimes\xi , B_{(k,i)}\otimes k\otimes\xi>
        +\ww<A_{(k,i)}\otimes k\otimes\xi,A_{(k,i)}\otimes k\otimes\xi>
    \Big)
    \notag \\
    &=
    \sum_{k\in\mathbb{Z}_n^+}\sum_{i=1}^{n-1}
    \int_M 
    \frac{\sin^2\vv<k,x>+\cos^2\vv<k,x>}{|k|^{2s+2}}
        \vv<k_i^{\bot}k^t\xi(x),k_i^{\bot}k^t\xi(x)>
    \,dx \notag \\
    &=
    \sum_{k\in\mathbb{Z}_n^+}\sum_{i=1}^{n-1}
    \int_M 
    \frac{\vv<k,\xi(x)>^2|k_i^{\bot}|^2}{|k|^{2s+2}}
    \,dx 
    =
    (n-1)\sum_{k\in\mathbb{Z}_n^+}
    \frac{1}{|k|^{2s}}
    \int_M 
    \vv<k,\xi(x)>^2
    \,dx 
    \notag 
\end{align}
where we use that $|k_i^{\bot}| = |k|$.

Equation~\eqref{e:dH} now yields
\begin{align}
    \dd{t}{}\mathcal{E}_t^s(\xi) 
    &=
    \by{1}{2}E\Big[-\ww<\xi_t, [u_t, \xi_t]>
    + 
    \nu^2(n-1)\sum_{k\in\mathbb{Z}_n^+}
    \frac{1}{|k|^{2s}}
    \int_M 
    \vv<k,\xi_t(x)>^2
    \,dx 
    - 
    \nu^2\sum \ww<\nabla f_{\alpha}^{\xi} , \nabla f_{\alpha}^{\xi}>\Big]
\end{align}
which certainly is not negative in general. We note that the non-dissipation of $\mathcal{E}^s$ is due to the line stretching terms~\eqref{e:line-stretching}, which are precisely the terms needed to make the Kelvin Circulation Theorem~\ref{prop:KC} hold.

\subsection{Dissipation}\label{sec:diss}
Let us now perform the same analysis as in Section~\ref{sec:non-diss}, but with respect to the formulation~\eqref{e:sde3}:
Consider, with $\nu=\sqrt{\by{2\eta}{c^s}}$ and $u_0\in\gu_0$ as above,  the mean-field equation
\begin{align}\label{e:sde1}
 \delta \wt{\zeta}_t 
 &= - \Big(P\nabla_{u_t}\wt{\zeta}_t + \eta\nabla\d\,\wt{\zeta}_t \Big)\,\delta t
    - \nu \sum \nabla_{X_{\alpha}}\wt{\zeta}_t \,\delta W_t^{\alpha}\\
 \wt{\zeta}_0 &= u_0  \notag \\
 u_t &= E[\wt{\zeta}_t] \notag
\end{align}
This is now an equation in $\gu$. Its solutions are not necessarily divergence free.  

\begin{theorem}[\cite{H17}]\label{thm:mainH17}
If $\wt{\zeta}_t\in\gu$ is a strong solution to \eqref{e:sde1} on $[0,T]$ such that $\wt{\zeta}_0=u_0\in\gu_0$ then $u(t,x) = E[\ev_x(\wt{\zeta}_t)]$ satisfies the Navier-Stokes equation \eqref{e:nse} with $u(0,.)=u_0$ for $t\in[0,T]$.
Further, the Ito version of \eqref{e:sde1} is
\begin{equation}\label{e:sde2}
 d \wt{\zeta}_t 
 = \Big(-P\nabla_{u_t}\wt{\zeta}_t - \eta\nabla\d\,\wt{\zeta}_t + \eta\Delta\wt{\zeta}_t \Big)\,dt
    - \nu \sum \nabla_{X_{\alpha}}\wt{\zeta}_t \, dW_t^{\alpha}.
\end{equation}
\end{theorem} 

In order to obtain a divergence free solution, we modify Equation~\eqref{e:sde2} as
\begin{equation}\label{e:sde3}
   d \zeta_t 
 =  (-P\nabla_{u_t}\zeta_t + \eta\Delta\zeta_t )\,dt
    - \nu \sum P\nabla_{X_{\alpha}}\zeta_t \, dW_t^{\alpha}
\end{equation}
with $u_t=E[\zeta_t]$ and the same initial condition $\zeta_0=u_0$.
Now we have an equation in $\gu_0$, such that 
$
\d\,\zeta = 0
$
and $u_t$ 
still solves the Navier-Stokes Equation~\eqref{e:nse}.
Indeed, if $\wt{\zeta}_t$ is as in Theorem~\ref{thm:mainH17}, then $\zeta_t = P\wt{\zeta}_t$ and $u_t = E[\wt{\zeta}_t] = PE[\wt{\zeta}_t] = E[\zeta_t]$.

\begin{proposition}\label{prop:EnergyDiss}
Assume $\zeta_t$ is a solution of \eqref{e:sde3}. 
Let $P\nabla_u\zeta = \nabla_u\zeta + \nabla p$, $P\nabla_{X_{\alpha}}\zeta = \nabla_{X_{\alpha}}\zeta + \nabla q^{\alpha}$
and $P\nabla_{X_{\alpha}}u = \nabla_{X_{\alpha}}u + \nabla q_u^{\alpha}$.
Let $\mathcal{E}^s_t = \by{1}{2}E[\ww<\zeta_t,\zeta_t>]$.
\begin{enumerate}[\up (1)]
\item
  $\mathcal{E}^s_t = E[\mathcal{H}_0(\zeta_t)]$ is dissipative:
  $\dd{t}{}\mathcal{E}_t^s = -\by{\nu^2}{2}\sum_{\alpha}\ww<\nabla q^{\alpha},\nabla q^{\alpha}>
  \le -\by{\nu^2}{2}\sum_{\alpha}\ww<\nabla q_u^{\alpha},\nabla q_u^{\alpha}> < 0$,
  for $\nu>0$.
\item
  $\mathcal{E}^s_t\ge \by{1}{2}\ww<E[\zeta_t],E[\zeta_t]> = \mathcal{E}_t^d$.
\item
  $\mathcal{E}^d_t$ is monotone decreasing in $t$.
\end{enumerate}
\end{proposition}

In particular, (deterministic) energy dissipation for the Navier-Stokes equation follows from the stochastic formulation~\eqref{e:sde3}, without using the PDE setting. 

\begin{proof}
Ad (1).
Observe that
\[
\ww<\zeta, P\nabla_u\zeta>
= \ww<\zeta, \nabla_u\zeta + \nabla p>
= \int_M\Big(\by{1}{2}\d(\vv<\zeta,\zeta>u)+\d(p\zeta)\Big)\,dx
= 0
\]
and
\begin{align*}
  &\nu^2\ww<P\nabla_{X_{\alpha}}\zeta,P\nabla_{X_{\alpha}}\zeta>
  = \nu^2\ww<\nabla_{X_{\alpha}}\zeta, \nabla_{X_{\alpha}}\zeta + \nabla q^{\alpha}>\\
  &= \nu^2\int_M\Big(\d\Big(\vv<\zeta,\nabla_{X_{\alpha}}>X_{\alpha}\Big)
  - \vv<\zeta,\nabla_{X_{\alpha}}\nabla_{X_{\alpha}}\zeta>\Big)\,dx
  +\nu^2\int_M\Big(\d\Big(q_{\alpha}\nabla_{X_{\alpha}}\zeta\Big)
  - q_{\alpha}\d\Big(\nabla_{X_{\alpha}}\zeta\Big)\Big)\,dx \\
  &= -2\eta\ww<\zeta,\Delta\zeta> + \nu^2\ww<q_{\alpha},\Delta q_{\alpha}>.
\end{align*}
It follows that
\begin{align}
\label{e:diss1}
  &d\ww<\zeta_t,\zeta_t>
  = 2\ww<\zeta_t,d\zeta_t>
  + \sum_{\alpha}\ww<P\nabla_{X_{\alpha}}\zeta,P\nabla_{X_{\alpha}}\zeta>dt  \\
  &=
  2\ww<\zeta_t, (-P\nabla_{u_t}\zeta_t + \eta\Delta\zeta_t)dt - \sum_{\alpha}P\nabla_{X_{\alpha}}\zeta\,dW_t^{\alpha}>
  - 2\eta\ww<\zeta_t,\Delta\zeta_t> dt
  - \nu^2\sum_{\alpha}\ww<\nabla q^{\alpha},\nabla q^{\alpha}>dt
  \notag
\end{align}
and, therefore, $d\ww<\zeta_t,\zeta_t> = - \nu^2\sum_{\alpha}\ww<\nabla q^{\alpha},\nabla q^{\alpha}>dt$ as-well as
\[
\dd{t}{}\mathcal{E}_t^s
= -\by{\nu^2}{2}\sum_{\alpha}E[\ww<\nabla q^{\alpha},\nabla q^{\alpha}>]. 
\]
Now, the Jensen inequality implies that 
\[
    E[\ww<\nabla q^{\alpha},\nabla q^{\alpha}>] 
    \ge \ww<E[\nabla q^{\alpha}],\nabla E[q^{\alpha}]> 
    = \ww<\nabla q_u^{\alpha},\nabla q_u^{\alpha}>
\]
which is non-zero since $\d\,\nabla_{X_{\alpha}}u = \tr(u'\otimes X_{\alpha}')\neq0$ for non-trivial solutions $u$ and general $\alpha$.

Ad (2).
This follows from the Jensen inequality.

Ad (3).
This has to be proved separately, since being bounded by a decreasing function does not imply being (locally) decreasing. 
Let $t_0<t_1$ (such that solutions $\zeta$ and $u$ exist on an interval containing these two points).
We can condition \eqref{e:sde3} to start at time $t_0$ with initial condition $\zeta_{t_0} = u_{t_0}$. Then $\mathcal{E}^d_{t_0}=\mathcal{E}^s_{t_0} > \mathcal{E}^s_{t_1}\ge\mathcal{E}^d_{t_1}$.
\end{proof}

We remark that the dissipation of the stochastic energy $\mathcal{E}^s$ follows not, as in the deterministic case, from the Laplacian in the drift, but from the quadratic variation of the martingale part in \eqref{e:sde3}.

One may ask whether $\mathcal{E}^s$ decreases as quickly as $\mathcal{E}^d$. To this end, note that $\by{\nu^2}{2}$ equals $\eta$ up to the $c^s$-factor, and
\begin{align*}
\ww<\nabla q_u^{\alpha},\nabla q_u^{\alpha}>
&= \ww<q_u^{\alpha},\d\,\nabla_{X{\alpha}}u> 
= \ww<\Delta^{-1}(\tr(u'\otimes X_{\alpha}'), \tr(u'\otimes X_{\alpha}')>.
\end{align*}
Hence the stochastic energy dissipates at a rate that is proportional to the $L^2$-square of $u'$, just as in the deterministic case.

\section{Interacting particle system (IPS)}\label{sec:IPS}
Numerical algorithms for the  simulation of mean field equations can be devised by means of IPS approximations. This is true, both, in the finite and infinite dimensional setting. See \cite{AD95,delmoral,Gov}.  
The mean field formulation of \cite{CI05} has been used in \cite{IM08,IN11,NS17} to derive and study Lagrangian formulations of associated interacting particle systems.
One of the key issues in these articles is the fact that the energy of the approximating particle system does not completely dissipate (even if the solution is defined globally in time). Hence \cite{NS17} use a resetting technique, inspired by \cite{IN11}, in their approximation scheme. 

Below we use equations \eqref{e:sdeP-Ito} and \eqref{e:sde3} to discuss our version of an approximating interacting particle system. 

\subsection{Version 1 -- \eqref{e:sdeP-Ito}}
Fix a large integer $N$.
The approximating IPS for the mean field Equation~\eqref{e:ham_mf2} is \eqref{e:ham_ips3}.
The Ito version of \eqref{e:ham_ips3} is 
\begin{align}\label{e:ips1}
    d\xi^i
    &= -P\Big(\nabla_{u^N}\xi^i + (u^N)'\otimes\xi^i\Big)\,d t
            +\eta\Delta\xi^i\,d t
       - \nu\sum P\Big(\nabla_{X_{\alpha}}\xi^i 
            + X_{\alpha}'\otimes\xi^i\Big)\,d W^{\alpha,i} \\
    u^N 
    &= \by{1}{N}\sum\xi^i \notag\\
    \xi^i(0)
    &= u_0
    \notag 
\end{align}
where $W^{i,\alpha}$ is a sequence of pairwise independent Brownian motions and $i=1,\dots,N$. Equation~\eqref{e:ips1} is therefore the approximating IPS for the mean field system \eqref{e:sdeP-Ito}. 
We have $\d(\xi^i)=0=\d(u^N)$. 
Note that the empirical average $u^N_t$ is a stochastic process.

Let $R_t$ denote the martingale part of $\mathcal{H}_0(u^N_t) = \by{1}{2}\ww<u^N_t,u^N_t>$. 
Equation~\eqref{e:crossterms} applied to $u^N$ yields
\begin{align}
\label{e:version1}
    &d\mathcal{H}_0(u^N) - dR\\
    &= 
    \ww<-P\Big(\nabla_{u^N}u^N + (u^N)'\otimes u^N\Big) + \eta\Delta u^N, u^N>\,dt
    + \by{\nu^2}{2}\sum\ww<\hat{Y}_{\alpha}(u^N),\hat{Y}_{\alpha}(u^N)> \,dt \notag\\
    &=
    \ww<u^N,\eta\Delta u^N>\,dt 
    + \by{\nu^2}{2N^2}\sum_{i}\Big(
    \ww<-c^s\Delta\xi^i,\xi^i>
    + \sum\ww<X_{\alpha}' \otimes\xi^i, X_{\alpha}' \otimes\xi^i>
    - \sum\ww<\nabla f_{\alpha}^i,\nabla f_{\alpha}^i> 
    \Big) \,dt \notag\\
    &=
    -\by{\eta}{N^2}\sum_{i\neq j} \ww<\nabla\xi^i,\nabla\xi^j >\,dt 
    + \by{\nu^2}{2N^2}\sum_{\alpha,i}\Big(
     \ww<X_{\alpha}' \otimes\xi^i, X_{\alpha}' \otimes\xi^i>
    - \ww<\nabla f_{\alpha}^i,\nabla f_{\alpha}^i> 
    \Big)\,dt \notag
\end{align}
where $\Delta f_{\alpha}^i = -\d(\nabla_{X_{\alpha}}\xi^i + X_{\alpha}'\otimes\xi^i)$.
The term $-\by{\eta}{N^2}\sum_{i\neq j} \ww<\nabla\xi^i,\nabla\xi^j >$ corresponds to \cite[Equation~(2.8)]{NS17}. The other two terms are due to the fact that the $X_{\alpha}$ are non-constant, which is a difference between \eqref{e:ips1} and \cite[Lemma~4.2]{IM08} that is due to the infinite dimensional approach. An explicit formula for the middle term of the right hand side is given in \eqref{e:X_alpha}.

\subsection{Version 2 -- \eqref{e:sde3}}
Consider now the approximating IPS for Equation~\eqref{e:sde3}: 
\begin{align}\label{e:ips2}
    \delta\zeta^i
    &= -P\nabla_{u^N}\zeta^i \,dt
        +\eta\Delta\zeta^i\,dt
       - \nu\sum P\nabla_{X_{\alpha}}\zeta^i \,dW^{i,\alpha} \\
    u^N 
    &= \by{1}{N}\sum\zeta^i \notag\\
    \zeta^i(0)
    &= u_0
    \notag 
\end{align}
where $W^{i,\alpha}$ is a sequence of pairwise independent Brownian motions and $i=1,\dots,N$.
Again we have $\d(\zeta^i)=0=\d(u^N)$. 
Let $R_t$ denote the martingale part of $\mathcal{H}_0(u^N_t) = \by{1}{2}\ww<u^N_t,u^N_t>$. The calculation~\revise{\eqref{e:diss1}} now implies 
\begin{align}
\label{e:version2}
   d\mathcal{H}_0(u^N) - dR
    &=
    -\by{\eta}{N^2}\sum_{i\neq j} \ww<\nabla\zeta^i,\nabla\zeta^j >\,dt 
    - \by{\nu^2}{2N^2}\sum_{\alpha,i}
      \ww<\nabla f_{\alpha}^i,\nabla f_{\alpha}^i> 
    \,dt 
\end{align}
where  
$\Delta f_{\alpha}^m = -\d(\nabla_{X_{\alpha}}\revise{\zeta}^m)$.

Both formulations, \eqref{e:version1} and \eqref{e:version2}, involve cross terms of the type $\sum_{i\neq j} \ww<\nabla\zeta^i,\nabla\zeta^j >$. This leads to the non-dissipation of the expectation $E[\mathcal{H}_0(u_t^N)]$ that has been observed in \cite{IM08,IN11,NS17}. 

\section{Comparison with other approaches}\label{sec:comp}

\subsection{Stochastic Lagrangian least action principle -- Cipriano and Cruzeiro \cite{CC07,C}}\label{sec:least-action}
Stochastic variational principles as a means to characterize solutions to the Navier-Stokes equation have already been considered by \cite{Y83}.
In \cite{CC07,C} solutions of the Navier-Stokes equation are characterized as solutions to a stochastic variational principle defined on the  group of volume preserving homeomorphisms $G^0_0$.
They consider 
Lagrangian curves of the form 
\begin{equation*}
\tag{\cite[Equ.~(3.1)]{CC07}}
    \delta g^u_t = TR^{g^u_t}
    \Big(u_t \,\delta t
     + \nu X_{\alpha}\,\delta W_t^{\alpha}\Big), \;
    g^u_0 = e 
\end{equation*}
where $u_t$ is a time dependent vector field, and prove that $u$ solves the Navier-Stokes equation if and only if $g^u_t$ is a solution to a stochastic variational principle (\cite[Theorem~4.1]{CC07}). Note that  \cite[Equ.~(3.1)]{CC07} coincides with our Equation~\eqref{e:ham_mf1}. By Theorem~\ref{thm:main}, $u$ solves the Navier-Stokes equation if and only if $(g_t^u,\Ad^{\top}(g^u_t).u_u)$ is a solution to the Hamiltonian  Equation~\eqref{e:ham_mf3}. Therefore, the stochastic variational principle of \cite{CC07} is equivalent to the Hamiltonian mean field formulation \ref{e:sdeP}. However, we remark that the regularity assumptions of Theorem~\ref{thm:main} are stronger than those of  \cite[Theorem~4.1]{CC07}. Thus we have proved this equivalence only under the stronger regularity assumptions on the solutions $u$ and $g^u$.

\subsection{Comparison with Constantin and Iyer~\cite{CI05}}\label{sec:CI}
In \cite{CI05} the stochastic vector $w = A'\otimes(u_0\circ A)$ is defined; in this definition $A=X^{-1}$, where $X=X(t,x)$ is the Lagrangian flow map, i.e.\ a volume preserving diffeomorphism for each $t$, and the divergence free vector field $u_0$ is the initial condition. 
As before, we make use of the notation 
$
 A'\otimes v = \sum(\del_i A^j)v^j e_i = (\nabla^t A)v
$
where $(e_i)$ is the standard basis in $\R^3$.
Let $B$ denote three-dimensional Brownian motion. 
In \cite[Theorem~2.2]{CI05} it is assumed that the pair $(X,u)$ satisfies
the mean field Ito SDE
\begin{equation*}
    dX = u\,dt + \sqrt{2\eta}\,dB,\;
    u = E[Pw],\;
    X(0,x) = x
\end{equation*}
and concluded that $u$ is a solution of the incompressible Navier-Stokes equation with initial data $u_0$ and viscosity coefficient $\eta$.
In order to do so, they
derive the system
\begin{equation*}
\tag{\cite[Equation~4.5]{CI05}}
dw_t = \Big(-\nabla_{u_t} w_t + \eta \Delta w_t - u_t'\otimes w_t\Big)\,dt + \sqrt{2\eta}\nabla w_t\,d\retwo{B}_t.
\end{equation*}
Now, \cite{CI05} construct the mean Eulerian velocity as $u = PE[w]$. 

It is further shown that a  Kelvin Circulation Theorem holds along stochastic Lagrangian paths:
\begin{equation*}
    \tag{\cite[Prop.~2.9]{CI05}}
    \oint_{X_t(C)}Pw_t = \oint_C u_0
\end{equation*}
where $C$ is a closed curve in $M$. 

The term $u'\otimes\xi$ in \eqref{e:sdeP} is also present in \cite[Equation~(4.5)]{CI05} and it is this term which makes the Kelvin Circulation Theorem work. In our approach we also have to include the $X_{\alpha}'\otimes\xi$ which is not needed in the finite dimensional setting of \cite{CI05}. Thus  Equation~\eqref{e:sdeP-Ito} is an infinite dimensional analogue of \cite[Equation~(4.5)]{CI05}.

\subsection{Comparison with Holm et al.~\cite{Holm15,CFH17}}\label{sec:Holm}
Consider vector fields $u_t$ and $\wt{X}_{\alpha}$ on $\mathbb{T}^3$ where  $u_t$ is time dependent. Let, as above, $W^{\alpha}$ be a sequence of mutually independent Brownian motions. Consider the Stratonovich SDE
\begin{equation*}
\tag{\cite[Equ.~(2.29)]{CFH17}}
\delta y_t = u_t\,\delta t + \nu\sum \widetilde{X}_{\alpha} \,\delta W_t^{\alpha}
\end{equation*}
where we use the notation $y_t$ to be consistent with \cite{CFH17} and add the parameter $\nu$. 
We remark that the $\widetilde{X}_{\alpha}$ correspond to the deterministic vector fields $\xi_i$ in \cite[Equ.~(2.29)]{CFH17}, while we use $\xi_t$ in this paper to denote the stochastic \revise{momentum} field.

Let $dx$ denote the standard volume element in $\revise{\mathbb{T}}^3$. Then \cite{CFH17} use a stochastic version of Newton's second law to derive the $3D$ \emph{stochastic Euler equation}
\begin{equation*}
 \tag{\cite[Equ.~(2.40)]{CFH17}}
 (\delta + \ad^*(\delta y_t)).\psi_t\otimes dx = -dp_t\otimes dx \,\delta t. 
\end{equation*}
where $\psi_t\otimes dx$ is now a stochastic curve of one-form densities and $p_t$ is the pressure. For $\eta\in\gu$ we have the dual pairing $\vv<\psi_t\otimes dx,\eta> = \int\vv<\psi_t(x),\eta(x)>\,dx$, and $\ad^*$ is defined with respect to this pairing as in Section~\ref{sec:InOp}. Since $dx$ is constant, we can reformulate the stochastic Euler equation as an equation in $\gu_0^*$:
\begin{equation}
    \label{e:stoch_euler_equ}
    (\delta + \ad^*(\delta y_t)).[\psi_t] = 0
\end{equation}
where we use again the notation from Section~\ref{sec:InOp}. 

Assume that $\wt{X}_{\alpha} = X_{\alpha}$. 
Note that this means that $\delta y_t = \delta g_t \circ g_t^{-1}$ in the terminology of \eqref{e:ham_mf1}.
Therefore, under these assumptions, the following are equivalent for a stochastic process $\zeta_t$ in $\gu_0$ with $E[\zeta_t] = v_t$: 
\begin{enumerate}[\up (1)]
    \item 
    The stochastic Newton law \cite[Equ.~(2.36)]{CFH17}, with constant density $\rho=1$ and where the only force on the fluid is the gradient of the pressure, holds as a one-form relation for $\zeta_t^{\flat}$.
    \item
    The stochastic Euler equation \cite[Equ.~(2.40)]{CFH17} holds for $\zeta_t^{\flat}\otimes dx$ in the space of one-form densities.
    \item
    Equation~\eqref{e:ham_mf2} holds with respect to $u$: 
    $\delta \zeta_t
    = \ad^{\top}(u_t).\zeta_t\,\delta t
    + \nu\sum_{\alpha}\ad^{\top}(X_{\alpha}).\zeta_t\,\delta W_t^{\alpha}$.
    \item
    $\zeta_t = \Ad^{\top}(g).u_0$ is the stochastic \revise{momentum} along $g_t$.
\end{enumerate}
Indeed, the equivalence of (1) and (2) is shown by \cite[Equ.~(2.37)]{CFH17}. Equivalence of (2) and (3) follows from the reformulation \eqref{e:stoch_euler_equ} and the identity $\check{\mu}\circ\ad^{\top} = \ad^*\circ\check{\mu}$. Equivalence of (3) and (4) is Equation~\eqref{e:proof-main}.

If (1) - (4) hold then, by Corollary~\ref{cor:main}, $u$ is a solution to the Navier-Stokes Equation~\eqref{e:nse} if and only if $u=v$. 
This is the case if and only if (2) and (3) are mean field equations, where (3) coincides with the system~\eqref{e:sdeP}. 

In particular, Proposition~\ref{prop:KC} can also be seen as a direct consequence of the stochastic Kelvin Circulation Theorem in \cite{Holm15}.

\subsection{Comparison with \cite{H17}}\label{sec:H17}
Equations \eqref{e:sde1} and \eqref{e:sdeP} are, at first sight, very similar, and so are the respective conclusions of Theorems~\ref{thm:mainH17} and \ref{thm:main}. However,
carrying out the calculation \eqref{e:kct1} with respect to $\zeta_t$, a solution of \eqref{e:sde1}, immediately shows that the line stretching terms $u'\otimes\zeta$ and $X'_{\alpha}\otimes\zeta$ do not have a counterpart in $\delta \zeta$, whence there is no cancellation. We conclude that the Kelvin Circulation Theorem does not hold for $\zeta$. 
Furthermore, \eqref{e:sde1} is an equation in $\gu$. Thus, from a structural point of view, \eqref{e:sde1} and \eqref{e:sdeP} are very different. 

In \cite{Holm02} it is argued that conservation of circulation is a quality criterion for any proposed model of fluid mechanics. Given Proposition~\ref{prop:kct_gen}, we can extend this argument as follows: The natural phase space of incompressible fluid mechanics is $G_0\times\gu_0$. If the model is Hamiltonian with respect to a (possibly vector valued) function $\mathcal{H}$, then $\mathcal{H}$ should be invariant under the particle relabeling symmetry, which is given by right action of $G_0$ on itself. Thus, by Proposition~\ref{prop:kct_gen}, if a model does not satisfy the Kelvin Circulation Theorem it cannot be a Hamiltonian system. This is why we prefer \eqref{e:sdeP} compared to \eqref{e:sde1}. In fact, the Hamiltonian structure is given by \eqref{e:ham_mf3}. We call this structure the \emph{Hamiltonian approach}.  

Equation~\eqref{e:sde1} was derived in \cite[Sec.~1]{H17} by using a stochastic version of the material derivative along Lagrangian paths. While this is mathematically correct, the resulting structural properties and physics are different, as argued in the previous paragraph. In this regard it is worth noticing that this \emph{material derivative approach} and the Hamiltonian approach are indeed equivalent in the deterministic case: This follows immediately from the observation that the deterministic version of the line stretching terms \eqref{e:line-stretching} is $u'\otimes u = \by{1}{2}d\vv<u,u>$, which vanishes upon integration over closed loops. Thus the two approaches differ, in the deterministic case, by a full differential, and this difference is absorbed by the pressure term. See also Section~\ref{sec:memin}.

\subsection{Comparison with M\'emin et al.~\cite{Mem14,RMC17}}\label{sec:memin}
The approach of \cite{Mem14,RMC17} is also based on the material derivative approach. 
Analogously to Equation~\eqref{e:ham_mf1}, they assume Lagrangian particle trajectories given by an Ito SDE of the form
\begin{equation*}
\tag{\cite[Equ.~(5)]{RMC17}}
    dX_t 
    = w(t,X_t)\,dt + \sigma(t,X_t)\,dB_t
\end{equation*}
where $B_t$ is $n$-dimensional Brownian motion. 
However, their material derivative approach is not derived from  a direct stochastic perturbation of the deterministic argument as in \cite[Sec.~1]{H17}, but from a stochastic Reynold's Transport Theorem applied to scalar quantities. Since the transformation of scalar quantities does not involve line stretching terms, the difference between the Hamiltonian and the material derivative approach disappears, as explained in Section~\ref{sec:H17}. 
Indeed, various versions of a stochastic Navier-Stokes equation are derived in this setting (\cite[Equ.~(43)]{Mem14} and \cite[Equ.~(38)]{RMC17}). 
Compared to the Hamiltonian  approach~\eqref{e:sdeP} there are the following differences:
\begin{enumerate}[\up (1)]
    \item
    The Kelvin Circulation Theorem does not hold for \cite[Equ.~(43)]{Mem14} or \cite[Equ.~(38)]{RMC17}. This is due to the same reasons that also apply to the model \eqref{e:sde1}, as explained in Section~\ref{sec:H17}. 
    \item 
    \cite[Sec.~5]{Mem14} uses a stochastic version of the stress tensor to model shear forces, whereas \eqref{e:sdeP} is derived without the stress tensor.  
    \item 
    Energy is conserved exactly along stochastic paths for the model of \cite{Mem14,RMC17}.
    However, this does not hold for Equation~\eqref{e:sdeP}, due to the line stretching terms as shown by \eqref{e:non-cons}. Nor does it hold for \eqref{e:sde1} or \eqref{e:sde3}, due to the non-solenoidal character of $\nabla_{X_{\alpha}}\zeta$ as shown by \eqref{e:diss1}.
\end{enumerate}

\section{Conclusions}\label{sec:conc}
We have constructed a Hamiltonian interacting particle system~\eqref{e:ham_ips1} for the description of an ensemble of fluid particles from the assumption that the energy separates into a deterministic and a stochastic part. The deterministic component is  responsible for the overall drift of the ensemble, while the stochastic part models the molecular transfer of momentum between fluid layers of different velocities. 
It is this molecular transfer that is responsible for shear forces which are thus accounted for in the interacting particle system (IPS). 
In the continuum limit, as the number of particles goes to infinity, this IPS yields a mean field \revise{SDE}~\eqref{e:sdeP}. The mean field system has the following properties: 
\begin{itemize}
    \item 
    It is equivalent to a stochastic Hamiltonian system~\eqref{e:ham_mf3} with respect to the natural symplectic structure on the phase space $G_0\times\gu_0$ of incompressible fluid mechanics.
    \item
    Averaging over solutions of \eqref{e:sdeP}, which are \emph{stochastic \revise{momenta identified with elements in velocity space}}, implies a solution to the Navier-Stokes Equation~\eqref{e:sdeP}. 
    Moreover, any Gaussian perturbation~\eqref{e:ham_mf1} of 
    a Lagrangian trajectory corresponding to a solution of the Navier-Stokes equation yields a mean-field system of the form~\eqref{e:sdeP}, and the mean field coincides with the original solution to the Navier-Stokes equation. (Theorem~\ref{thm:main})
    \item
    Consider the average over stochastic \revise{momenta} along a Gaussian perturbation~\eqref{e:ham_mf1} of a divergence-free time dependent vector field. 
    By Corollary~\ref{cor:main}, solutions to the Navier-Stokes equation can be characterized as a fixed point of this assignment. 
    \item
    There are the following analogies to the case of ideal fluids:
    \begin{itemize}
    \item
    The Hamiltonian structure implies conservation of circulation along stochastic paths in the phase space where one jointly considers configuration and \revise{momentum} variables. Since these variables cannot be separated in the formulation of the Kelvin Circulation Theorem, this does not imply conservation of circulation for the expectation of the stochastic \revise{momentum} -- i.e., for a solution of the Navier-Stokes equation.  
    \item
    The Navier-Stokes equation is derived from a (stochastic) Hamiltonian system without use of a stress tensor. 
    \end{itemize}
    \item
    Energy is neither conserved nor dissipating (Section~\ref{sec:non-diss}). However, the approach of \cite{H17} can be slightly modified to yield a system which is dissipative (Section~\ref{sec:diss}).
\end{itemize}
Finally, we have contrasted the model~\eqref{e:sdeP} with existing approaches from the literature and established several connections.
\revise{Our model should also be compared with \cite{Holm02a,HT12} where (deterministic) advection of fluid microstructure is studied. The \emph{Kinematic Sweeping Ansatz} of \cite{HT12} is to consider deterministic turbulence parameters which are swept along a mean flow, while \eqref{e:sdeP} is a model for stochastic particles along a mean flow. A detailed comparison of these approaches is a task for future research.}

\end{document}